\documentclass[letterpaper, 10 pt, journal, twoside]{ieeeconf}
\IEEEoverridecommandlockouts   
\usepackage{cite}

\let\labelindent\relax
\usepackage{enumitem}

\usepackage{amsthm}
\usepackage[cmex10]{amsmath}
\newtheorem{theorem}{Theorem}
\newtheorem{lemma}{Lemma}
\newtheorem{definition}{Definition}
\newtheorem{assumption}{Assumption}

\usepackage[T1]{fontenc}
\usepackage[utf8]{inputenc}
\makeatletter
\let\NAT@parse\undefined
\makeatother
\usepackage[hyperindex=true, %
  pdftitle={Non-Convex Feedback Optimization with Input and Output Constraints},%
  pdfauthor={Verena H\"{a}berle, Adrian Hauswirth, Lukas Ortmann, Saverio Bolognani, Florian Dörfler},%
  colorlinks=true,%
  pagebackref=false,%
  plainpages=false,%
  pdfpagelabels,%
  linkcolor=black,%
  citecolor=black,%
  filecolor=black,%
  urlcolor=black%
]{hyperref} %

\usepackage[capitalize,nameinlink,compress]{cleveref}
\crefname{figure}{Figure}{Figures} 
\crefname{equation}{}{} 
\crefname{assumption}{Assumption}{Assumptions}

\usepackage{amssymb} 

\newcommand{\bbR}{\mathbb{R}}
\newcommand{\bbS}{\mathbb{S}}

\newcommand{\calU}{\mathcal{U}}
\newcommand{\calY}{\mathcal{Y}}
\newcommand{\calX}{\mathcal{X}}
\newcommand{\calC}{\mathcal{C}}
\newcommand{\calS}{\mathcal{S}}
\newcommand{\calT}{\mathcal{T}}

\newcommand{\linTildeU}{\tilde{\calU}_u}

\usepackage{graphicx} 

\title{\LARGE \bf Non-Convex Feedback Optimization with Input and Output Constraints}

\author{Verena H\"{a}berle, Adrian Hauswirth, Lukas Ortmann, Saverio Bolognani, Florian D\"{o}rfler
\thanks{The research leading to this work was supported in part by the Swiss Federal Office of Energy grant \#SI/501708 UNICORN.}
\thanks{The authors are with the Automatic Control Laboratory, ETH Z\"urich, Physikstrasse 3, 8092 Z\"urich, Switzerland.}
\thanks{Email: \{verenhae,hadrian,ortmannl,bsaverio,dorfler\}@ethz.ch.}
}

\begin{document}

\maketitle
\thispagestyle{empty}
\pagestyle{empty}

\begin{abstract}
In this paper, we present a novel control scheme for feedback optimization. That is, we propose a discrete-time controller that can steer a physical plant to the solution of a constrained optimization problem without numerically solving the problem. Our controller can be interpreted as a discretization of a continuous-time projected gradient flow. Compared to other schemes used for feedback optimization, such as saddle-point schemes or inexact penalty methods, our control approach combines several desirable properties: it asymptotically enforces constraints on the plant steady-state outputs, and temporary constraint violations can be easily quantified. Our scheme requires only reduced model information in the form of steady-state input-output sensitivities of the plant. Further, global convergence is guaranteed even for non-convex problems. Finally, our controller is straightforward to tune, since the step-size is the only tuning parameter.
\end{abstract}

\section{Introduction}\label{sec:introduction}
In recent years, the design of feedback controllers that steer the steady state of a physical plant to the solution of a constrained optimization problem has garnered significant interest both for its theoretical depth~\cite{jokicconstrainedsteadystateregulation2009,nelson_integral_2018,colombinoOnlineOptimizationFeedback2019, lawrence_optimal_2018} and potential applications. In particular, while the historic roots of \emph{feedback optimization} trace back to process control~\cite{garciaOptimaloperationintegrated1984,costelloModifierAdaptationConstrained2014} and communication networks~\cite{kellyRatecontrolcommunication1998, lowInternetcongestioncontrol2002}, recent efforts have centered around online optimization of power grids~\cite{li2015connecting,dallaneseOptimalPowerFlow2018,hauswirthOnlineoptimizationclosed2017, molzahnSurveyDistributedOptimization2017}.

We adopt the perspective that feedback optimization emerges as the interconnection of an optimization algorithm such as gradient descent (formulated as an open system) and a physical plant with well-defined steady-state behavior. This is in contrast, for example, to~\cite{lawrence_optimal_2018} that adopts an output-regulation viewpoint or extremum-seeking which is completely model-free~\cite{ariyurRealtimeoptimization2003}. In particular, for our purposes, we assume that the plant is stable with fast-decaying dynamics and the steady-state input-to-output map $y = h(u)$ is well-behaved (Fig.~\ref{fig:blockdiagramm}). This assumption is motivated by previous work on timescale separation in these setups~\cite{mentaStabilityDynamicFeedback2018, hauswirthTimescaleSeparationAutonomous2019}. 

The critical aspect of feedback optimization is that, instead of relying on a full optimization model, the algorithms take advantage of measurements of the system output. This entails that the system model $h$ does not need to be known explicitly, nor does it need to be evaluated numerically. Instead, only information about the steady-state sensitivities $\nabla h$ is required. This renders feedback optimization schemes inherently more robust against disturbances and uncertainties than ``feedforward'' numerical optimization.

A particular focus in feedback optimization is the incorporation of (unilateral) constraints on inputs and steady-state outputs. The former can be enforced directly by projection or by exploiting physical saturation and using anti-windup control~\cite{hauswirthImplementationProjectedDynamical2020, hauswirthAntiWindupApproximationsOblique2020}. Output constraints, however, cannot be enforced directly, especially when the map $h$ is unknown. Previous works either employ inexact penalty methods~\cite{tangRealTimeOptimalPower2017,hauswirthOnlineoptimizationclosed2017}, treating the output constraints as soft constraints, or saddle-point schemes~\cite{colombinoOnlineOptimizationFeedback2019,dallaneseOptimalPowerFlow2018,li2015connecting}, ensuring the constraints to be satisfied asymptotically. The latter, however, exhibit oscillatory behavior, are difficult to tune, and do not come with strong guarantees for non-convex problems~\cite{goebelStabilityrobustnesssaddlepoint2017,cherukuriRoleConvexitySaddlePoint2017}.

\begin{figure}[t] 
        \centering
        \includegraphics[scale=0.78]{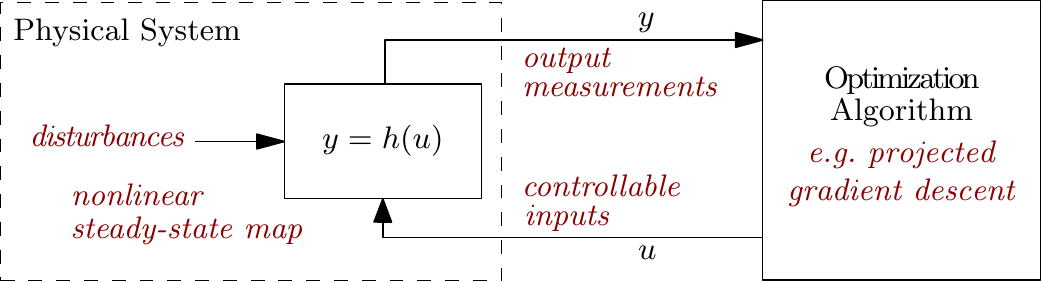}
        \caption{Block diagram of the feedback optimization setup.}
        \vspace{-3mm}
        \label{fig:blockdiagramm}
\end{figure}

In this paper, we present a new discrete-time controller to enforce output constraints in feedback optimization. Our scheme works by projecting gradient iterates onto a linearization of the feasible set around the current state and then applying them as set-points to the system.
For the main result, the global convergence of our scheme, we take inspiration from numerical algorithms like sequential quadratic programming (SQP)~\cite[Ch. 10]{sherali2006nonlinear},\cite[Ch. 18]{nocedal2006numerical},\cite{torrisi2018projected}. However, these algorithms require second-order information on $h$ or the exact knowledge of $h$ for line-searching, which is not always available in our setup. Hence, we cannot rely on line-search techniques and instead have to establish convergence for fixed step-sizes, ensuring that our optimization approach can be implemented as a time-invariant feedback controller.

Compared to saddle-point schemes, our controller demonstrates that the integration of a dual variable is not required to guarantee constraint satisfaction with zero asymptotic error. Further, our scheme exhibits a benign convergence behavior without oscillations, convergence is guaranteed for non-convex problems, and tuning is restricted to one parameter.

The rest of the paper is structured as follows:
In Section~\ref{sec:preliminaries}, we fix the notation and recall some preliminary technical results. In Section~\ref{sec:discrete_time_proj_grad}, we present our new discrete-time controller and state our main convergence result. The proof is laid out in Section~\ref{sec:proof_global_conv}. Finally, we give a numerical example in Section~\ref{sec:numerical_simulation} and discuss open questions in Section~\ref{sec:conclusion}. In the appendix, we indicate how our discrete-time controller is connected to a continuous-time projected gradient flow.

\section{Preliminaries}\label{sec:preliminaries}
For $\bbR^p$, $\langle\cdot,\cdot\rangle$ denotes the Euclidean inner product and $||\cdot||$ its induced 2-norm. The non-negative orthant of $\bbR^p$ is written as $\bbR^p_{\geq 0}$ and the set of symmetric positive definite matrices of size $p\times p$ is denoted by $\bbS^p_+$. Any $G \in \bbS^p_+$ induces a 2-norm defined as $||v||_G:=\sqrt{v^T G v}$ for all $v \in \bbR^p$. Given a set $\calC \subset \bbR^p$, a map $G: \calC \rightarrow\bbS^p_+$ is called a \emph{metric} on $\calC$.

For a continuously differentiable function $f:\bbR^p\rightarrow\bbR^q$, $\nabla f(x)\in\bbR^{q\times p}$ denotes the \emph{Jacobian} of $f$ at $x$. The map $f$ is \emph{globally $L$-Lipschitz continuous} if for all $x,y\in \bbR^p$ and some $L>0$, it holds that $||f(x)-f(y)||\leq L||x-y||$.

To establish our main convergence result in Section~\ref{sec:proof_global_conv}, we require the so-called \emph{Descent Lemma}~\cite[Prop A.24]{bertsekas1999nonlinear}:
\begin{lemma}\label{descent_lemma}
    Given a continuously differentiable function $f:\mathbb{R}^p\rightarrow\mathbb{R}$ with $L$-Lipschitz derivative $\nabla f$, for all $x, z \in \bbR^p$ it holds that $f(z)\leq f(x)+ \nabla f(x)(z - x) +\tfrac{L}{2}||z - x||^2$.
\end{lemma}

\subsection[Nonlinear Optimization Perturbation Analysis]{Nonlinear Optimization \& {Perturbation Analysis}}\label{sec:nonlinear_optimization}

In this paper, we often consider feasible sets of the form
\begin{align}\label{eq:cstr_set}
    \calX:=\{x\in\mathbb{R}^p\,|\,g(x)\leq 0 \} \, ,
\end{align} 
where $g:\mathbb{R}^p\rightarrow\mathbb{R}^q$ is continuously differentiable. Let $\mathbf{I}^\calX_x:=\{i\,|\,g_i(x)=0\}$ denote the \emph{index set of active inequality constraints of $\calX$ at~$x$} and let $\Bar{\mathbf{I}}^\calX_x:=\{i\,|\,g_i(x)<0\}$ be the index set of \emph{inactive} inequality constraints at $x$.

\begin{definition}[LICQ]\label{definition:LICQ_nonlinear_optimization_problem}
    Given $\calX \subset \bbR^p$ as in~\eqref{eq:cstr_set}, the \emph{linear independence constraint qualification} (LICQ) is said to hold at $x\in\calX$, if the matrix $\nabla g_{\mathbf{I}^\calX_x}(x)$ has rank $|\mathbf{I}^\calX_x|$.
\end{definition}

For a continuously differentiable function $\Psi:\mathbb{R}^p\rightarrow\mathbb{R}$ and $\calX \subset \bbR^p$ as in~\eqref{eq:cstr_set}, consider the constrained problem
\begin{align}\label{eq:constrained_optimization_problem}
  \underset{x}{\text{minimize}} \quad \Psi(x)  \quad
    \text{subject to} \quad x \in \calX \,.
\end{align}
The \emph{Lagrangian} of \eqref{eq:constrained_optimization_problem} is defined as $L(x,\mu):=\Psi(x)+\mu^Tg(x)$ for all $x \in \bbR^p$ and all \emph{Lagrange multipliers} $\mu\in\bbR^q_{\geq 0}$. Recall the \emph{first-order optimality (KKT) conditions} for~\eqref{eq:constrained_optimization_problem}:

\begin{theorem}\cite[Ch. 11.8]{luenberger1984linear}\label{theorem:KKT_condition_nonlinear_optimization_problem} If $x^\star\in\calX$ is a local solution of~\eqref{eq:constrained_optimization_problem} and LICQ holds at $x^\star$, there exists a unique $\mu^\star\in\bbR^q_{\geq 0}$ such that
$\nabla_x L(x^\star,\mu^\star)=0$ and $\mu^\star_i=0$ hold for all $i\in\mathbf{\Bar{I}}_{x^\star}^\calX$.
\end{theorem}
In particular, the LICQ assumption guarantees the uniqueness (and boundedness) of the dual multipliers $\mu^\star$~\cite{wachsmuth2013licq}.

For the analysis of our proposed controller in 
\cref{sec:discrete_time_proj_grad}, we need to consider \emph{parametric} problems of the form
\begin{equation}\label{eq:parametric_problem}
\begin{split}
    \underset{x}{\text{minimize}} \quad & \Psi(x, \varepsilon) \\
    \text{subject to} \quad & x \in \calX(\varepsilon) \,,
\end{split}
\end{equation}
where $\calX(\varepsilon):=\{x\in \mathbb{R}^p\,|\,g(x,\varepsilon)\leq0\}$, and $\Psi:\mathbb{R}^p\times\mathcal{T}\rightarrow \mathbb{R}$ and $g:\mathbb{R}^p\times\mathcal{T}\rightarrow{\mathbb{R}^q}$ are parametrized in~$\varepsilon \in \mathcal{T}\subset\bbR^r$. Hence, the Lagrangian of~\eqref{eq:parametric_problem} is defined as $L(x, \mu, \varepsilon) := \Psi(x, \varepsilon)+\mu^T g(x, \varepsilon)$ for all $x \in \bbR^p$, $\mu \in \bbR^q_{\geq 0}$, and all $\varepsilon \in \calT$.

To prove our main result in Section~\ref{sec:proof_global_conv}, we will require solutions of~\eqref{eq:parametric_problem} to be continuous as a function of $\varepsilon$. For this purpose, we use \cite[Thm 2.3.2]{jittorntrum1978sequential} which, for convex problems with strongly convex objective, simplifies to the following:
\begin{theorem}\label{thm:continuity_solution_map} Consider~\eqref{eq:parametric_problem} and assume that $\Psi$ and~$g$ are twice continuously differentiable in $x$, and that $\Psi, g, \nabla_x \Psi, \nabla_x g, \nabla^2_{xx} \Psi$, and $\nabla^2_{xx} g$ are continuous in $\varepsilon$. Furthermore, for all $\varepsilon \in \calT$, let 
\begin{itemize}
    \item $\Psi$ be strongly convex in $x$,
    \item  $\calX(\varepsilon)$ be non-empty and convex in $x$, and
    \item LICQ be satisfied for all $x \in \calX(\varepsilon)$.
\end{itemize}
Then, there exist continuous functions $x^\star: \calT \rightarrow \bbR^p$ and $\mu^\star: \calT \rightarrow \bbR^q_{\geq 0}$ such that $x^\star(\varepsilon)$ is the unique global optimizer of~\eqref{eq:parametric_problem} for all $\varepsilon \in \calT$ and $\mu^\star(\varepsilon)$ is its Lagrange multiplier.
\end{theorem}
\begin{proof}
We show that, under the given convexity assumptions, the requirements for \cite[Thm 2.3.2]{jittorntrum1978sequential} are met globally. Namely, by assumption, \eqref{eq:parametric_problem} is feasible for all $\varepsilon\in\mathcal{T}$ and LICQ holds for all $x\in\calX(\varepsilon)$ and all $\varepsilon\in\mathcal{T}$. Hence, by strong convexity of $\Psi$,~\eqref{eq:parametric_problem} admits a unique (global) optimizer for all $\varepsilon\in\calT$. Therefore, the solution map ${\varepsilon} \mapsto x^\star({\varepsilon})$ and $\varepsilon \mapsto \mu^\star(\varepsilon)$ are single-valued. Moreover, for all $\varepsilon \in \calT$, the KKT conditions (Theorem~\ref{theorem:KKT_condition_nonlinear_optimization_problem}) are satisfied and the \emph{second order sufficiency conditions} hold (trivially) by (strong) convexity. It then follows from \cite[Thm 2.3.2]{jittorntrum1978sequential} that $x^\star$ and $\mu^\star$ are continuous around every $\varepsilon \in \calT$ and hence on all of $\calT$.\end{proof}

\section[Problem Formulation \& Main Result]{{Problem Formulation \& Main Result}}
\label{sec:discrete_time_proj_grad}
We consider the problem of steering a physical plant to a steady state that solves a pre-specified constrained optimization problem.
We assume that the plant is described by a continuously differentiable nonlinear steady-state input-to-output map $h:\mathbb{R}^p\rightarrow\mathbb{R}^n$.

For simplicity, we consider separate constraints on the input $u \in \bbR^p$ and the output $y \in \bbR^n$ given by the polyhedra
\begin{align*}
    \mathcal{U}:=\{u\in\mathbb{R}^p\,|\,Au\leq b\}
    \quad \text{and} \quad
    \mathcal{Y}:=\{y\in\mathbb{R}^n\,|\,Cy\leq d\} \, ,
\end{align*}
where $A\in\mathbb{R}^{q\times p},b \in\mathbb{R}^q, C \in\mathbb{R}^{l\times n}$, and $d\in\mathbb{R}^l$.\footnote{For many applications, separable constraint sets are sufficiently expressive. Future work will address generalized  constraints on inputs and outputs.}

Given a continuously differentiable cost function $\Phi:\mathbb{R}^p \times \bbR^n \rightarrow\mathbb{R}$, we hence consider the problem
\begin{equation}
\begin{aligned}\label{eq:constrained_opt}
    \underset{u,y}{\text{minimize}} \quad & \Phi(u, y) \\
    \text{subject to} \quad & y = h(u) \\
                            & u \in \calU, \, y \in \calY \,.
\end{aligned}
\end{equation}
Note that the nonlinearity of $h$ and the non-convexity of $\Phi$ generally make it intractable to find a global solution of \eqref{eq:constrained_opt}.

In general, we assume that \eqref{eq:constrained_opt} is feasible:
\begin{assumption}\label{ass:basic_feas}
    The feasible set of~\eqref{eq:constrained_opt} is non-empty, i.e.,
    \begin{align}\label{eq:tilde_U}
        \tilde{\calU} := \calU \cap h^{-1}(\calY) = \{ u \, | \, A u \leq b, \, C h(u) \leq d \} \neq \emptyset .
    \end{align} 
\end{assumption}

Local solutions of problems of the form~\eqref{eq:constrained_opt} can be reliably computed using standard methods from numerical optimization \cite{sherali2006nonlinear,nocedal2006numerical} if all problem parameters are known.

In our context, however, the map $h$ representing a physical plant might be subject to disturbances and not fully known. Therefore, a precomputed ``feedforward'' solution based on an estimate of $h$ lacks robustness.

We hence consider a closed-loop \emph{feedback optimization} scheme to solve~\eqref{eq:constrained_opt}. Namely, we propose an integral feedback controller defined as
\begin{equation} \label{eq:feedbackupdate} 
    u^+ = u + \alpha\,\widehat{\sigma}_\alpha(u,y) \,  \quad y = h(u) \,, 
\end{equation}
where $\alpha>0$ is a fixed step-size, $y = h(u)$ is the measured system output, and $\widehat{\sigma}_\alpha:\bbR^p\times\bbR^n\rightarrow\bbR^p$ is defined as
\begin{align}
   \widehat{\sigma}_\alpha(u,y) := \arg \min_{w\in\bbR^p} \, &\| w + G^{-1}(u) H(u)^T \nabla \Phi (u,y)^T  \|^2_{G(u)}
    \nonumber\\ \label{eq:sigma_hat}
   \text{subject to} \quad &A (u+ \alpha w)\leq b \\ & {C (y+ \alpha\nabla h(u)w)\leq d}\nonumber\,,
\end{align}
where $H(u)^T:=\begin{bmatrix}\mathbb{I}_p&\nabla h(u)^T\end{bmatrix}$ and $G: \calU \rightarrow \bbS^p_+$ is a continuous metric on $\calU$. Note that the evaluation of~\eqref{eq:sigma_hat} is computationally tractable and does not require an explicit computation of $h$, as $y = h(u)$ is available by measurement. Instead, it is enough to know $\nabla h$ (or an approximation thereof).

Intuitively, $\hat{\sigma}_{\alpha}(u,y)$ in~\eqref{eq:sigma_hat} can be interpreted as the projection of the point $u-\alpha\, G^{-1}(u) \nabla \left( \Phi(u, h(u)) \right)$ onto a linearization of $\tilde{\calU}$ around $u$, evaluated at the measurement $y=h(u)$. Furthermore, as $\alpha \searrow 0^+$, we recover a continuous-time projected gradient operator. For a formal derivation of this fact, the reader is referred to the appendix.

Otherwise, the control law~\cref{eq:feedbackupdate,eq:sigma_hat} resembles standard numerical optimization algorithms, e.g., SQP~\cite[Ch. 10]{sherali2006nonlinear},\cite[Ch. 18]{nocedal2006numerical}, where a non-convex problem as in~\eqref{eq:constrained_opt} is solved via a sequence of convex QPs, each in turn being solved by any off-the-shelf convex optimization method~\cite{sherali2006nonlinear,nocedal2006numerical}. 

The main differences to standard numerical optimization algorithms are that, on the one hand,~\eqref{eq:feedbackupdate} does not rely on a numerical evaluation of $h$. On the other hand, we cannot use line-search techniques to force convergence. Instead, we need to consider fixed step-sizes which do not in general lead to convergence unless additional assumptions are made.

For~\eqref{eq:feedbackupdate} to be well-defined, we require the following:
\begin{assumption}\label{ass:linearizedLICQ}
    For all $u\in\calU$, the feasible set of~\eqref{eq:sigma_hat} defined as $\linTildeU := \{ w \, | \, A (u+ \alpha w)\leq b, \, C (y+ \alpha\nabla h(u)w)\leq d \}$ is non-empty and satisfies LICQ for all $w \in \linTildeU$.
\end{assumption}
In particular, \cref{ass:linearizedLICQ} guarantees that $\widehat{\sigma}_\alpha(u, y)$ is single-valued for all $u \in \calU$ and $y = h(u)$ since~\eqref{eq:sigma_hat} is a convex program with strongly convex objective. Further, the LICQ assumption will enable us to apply \cref{thm:continuity_solution_map}.

Generally, \cref{ass:linearizedLICQ} is common in the study of SQP~\cite[Ch. 10]{sherali2006nonlinear},\cite[Ch. 18]{nocedal2006numerical},~\cite{torrisi2018projected}. Providing sufficient conditions for \cref{ass:linearizedLICQ} to hold, are part of ongoing work. Nevertheless, it is known that for large classes of perturbed optimization problems, LICQ holds generically~\cite{spingarn1979generic}, which is in particular true for the type of problems in our envisioned applications to power systems~\cite{hauswirth2018generic}. 

Finally, we assume compactness of the plant inputs:
\begin{assumption}\label{ass:basic}
    For~\eqref{eq:constrained_opt}, the set $\calU$ is compact.
\end{assumption}
The assumption that $\calU$ is compact is motivated by the fact that physical plants can generally only handle bounded inputs. This seems logical, for instance, from the viewpoint that any reasonable physical signal has finite energy.

From a theoretical perspective, compactness of $\calU$ allows us to exploit the extreme value theorem, i.e., every continuous function attains a maximum on a compact set. Whether this assumption can be relaxed remains an open question.

Our following main result guarantees global convergence of~\eqref{eq:feedbackupdate} to the set of first-order optimal points of~\eqref{eq:constrained_opt} for a small enough, fixed step-size.

\begin{theorem}\label{MainTheorem_Linearized}
Under \cref{ass:basic_feas,ass:linearizedLICQ,ass:basic}, consider~\eqref{eq:constrained_opt} and assume that $\nabla{\Phi}$ and $\nabla h$ are globally Lipschitz on $\calU$. Then, there exists an $\alpha^\star > 0$ such that for every $\alpha < \alpha^\star$
\begin{enumerate}[label=(\roman*)]
    \item\label{it:cvg} the trajectory of~\eqref{eq:feedbackupdate} for any $u^0 \in\calU$ and $y=h(u)$ converges to the set of first-order optimal points of~\eqref{eq:constrained_opt}, and
    \item\label{it:stab} if $u^\star$ is an asymptotically stable equilibrium of~\eqref{eq:feedbackupdate}, it is a strict local minimum of $\Tilde{\Phi}(u) := \Phi(u, h(u))$ on $\Tilde{\mathcal{U}}$.
\end{enumerate}
\end{theorem}

\section{Proof of~\cref{MainTheorem_Linearized}}\label{sec:proof_global_conv}
To study the closed-loop behavior in the input coordinates, we substitute $y=h(u)$ and $\Tilde{\Phi}(u) = \Phi(u, h(u))$ in ~\eqref{eq:feedbackupdate} and~\eqref{eq:sigma_hat}, which can be equivalently expressed as
\begin{align} \label{eq:outproj_approx}
    u^+ = u + \alpha\, \sigma_\alpha (u),\,\,\,\,u\in\calU,
\end{align}
with $\sigma_\alpha: \bbR^p \rightarrow \mathbb{R}^p$ defined as 
\begin{subequations} \label{eq:sigma_operator}
    \begin{align}\label{eq:sigma_objective}
            \sigma_\alpha(u) := \arg \min_{w\in\bbR^p} \quad &\left\| w+G^{-1}(u) \nabla \Tilde\Phi(u)^T  \right\|_{G(u)}^2\\\label{eq:lin_set_u}
            \text{subject to} \quad &A (u+ \alpha w)\leq b \\\label{eq:lin_set_y} & C (h(u)+ \alpha\nabla h(u)w)\leq d \,.
    \end{align}
\end{subequations}
Note that $\nabla \tilde{\Phi}$ is Lipschitz since $\nabla h$ and $\nabla \Phi$ are Lipschitz by assumption.

To prove (i) in \cref{MainTheorem_Linearized} we will apply the following invariance principle for discrete-time systems \cite[Thm 6.3]{lasalle1976stability}: 
\begin{theorem}\label{thm:lasalle_discrete}
    Consider a discrete-time dynamical system $u^+ = T(u)$, where $T: \calS \rightarrow \calS$ is well-defined and continuous, and $\calS \subset \bbR^p$ is closed. Further, let $V: \calS \rightarrow \bbR$ be a continuous function such that $V(T(u)) \leq V(u)$ for all $u \in \calS$. 
    Let $\mathbf{u} = \{ u_0, u_1, u_2, \ldots \} \subset \calS$ be a bounded solution. Then, for some $r \in V(\calS)$, $\mathbf{u}$ converges to the non-empty set that is the largest invariant subset of $V^{-1}(r) \cap \calS \cap \{ u \, | \, V(T(u)) - V(u) = 0 \}$.
\end{theorem}

We first need to establish that the map $T(u) := u + \alpha\, \sigma_\alpha(u)$ is continuous in $u$ and $T(u) \in \calU$ for all $u \in \calU$.

\begin{lemma}\label{lem:closediteration}
    Under \cref{ass:basic_feas,ass:linearizedLICQ,ass:basic}, $u + \alpha\, \sigma_\alpha(u) \in \calU$ holds for all $u \in \calU$. 
\end{lemma}

\cref{lem:closediteration} follows since $\sigma_\alpha(u)$ satisfies~\eqref{eq:lin_set_u} by definition.

\begin{lemma}\label{lemma_linearized_output_projection_continuity}
Under \cref{ass:basic_feas,ass:linearizedLICQ,ass:basic}, $\sigma_\alpha$ and the map of associated Lagrange multipliers are continuous in $u$. \end{lemma}

\begin{proof}
{We can apply~\cref{thm:continuity_solution_map} to problem~\eqref{eq:sigma_operator}, which is parametrized in $u\in\calU$}: For all $u\in\calU$, the functions in~\eqref{eq:sigma_operator} and their first and second derivatives with respect to $w$ are continuous in $u$, the objective is strongly convex in $w$, the feasible set of~\eqref{eq:sigma_operator} is non-empty and convex in $w$, and LICQ is satisfied for all feasible $w$ (by \cref{ass:linearizedLICQ}). {By Theorem 2, $\sigma_\alpha$ and the Lagrange multipliers are continuous in $u$}.\end{proof}

Moreover, the constraint violation committed at every iteration of~\eqref{eq:outproj_approx} can be bounded.
\begin{lemma}\label{lemma:transient_bound_violations}
Assume that $C_i\nabla h$ is $\ell_i$-Lipschitz for all $i=1,...,l$ on $\calU$. Given the iteration~\eqref{eq:outproj_approx} and any $u\in\calU$, we have $C_{i}h(u^+)-d_{i} \leq \tfrac{\ell_i}{2}||\alpha \sigma_\alpha(u)||^2$. \end{lemma}
\begin{proof}
Using the Descent Lemma (\cref{descent_lemma}), the desired bound can obtained by inserting~\eqref{eq:lin_set_y} into
\begin{equation*}\label{eq:descent_lemma_Ch}
    C_{i}h(u^+)-C_{i}h(u)\leq\alpha C_{i}\nabla h(u)w+\tfrac{\ell_i}{2}||\alpha w||^2 \,. \qedhere
\end{equation*}
\end{proof}

\subsubsection*{Lyapunov Function}
Given $u\in\calU$, let $\mu_i^\star(u)$ be the Lagrange multiplier of \eqref{eq:sigma_operator} for the $i$th constraint of~\eqref{eq:lin_set_y} with $i=1,\ldots, l$. Since $\mu_i^\star(u)$ is continuous on the compact set $\calU$ (by \cref{lemma_linearized_output_projection_continuity}), there exists $\xi\geq \sup_{u \in \calU; i=1, \ldots, l} \{\mu^\star_i(u)\}$. We may consider the function $V:\mathbb{R}^p\rightarrow\mathbb{R}$, defined as
\begin{equation}
    V(u)=\Tilde{\Phi}(u)+\xi \left[ \textstyle\sum_{i=1}^l{\mathrm{max}\{0,C_{i}h(u)-d_i\}} \right],
    \label{eq:definition_V}
\end{equation}
which we show to be non-increasing along solutions of~\eqref{eq:outproj_approx}.

To prove this claim, note that~\eqref{eq:sigma_operator} is equivalent to solving
\begin{subequations}
    \begin{align}
            \arg \min_{w\in\bbR^p} \quad
            &\tfrac{\alpha}{2}{w}^TG(u)w+\alpha\nabla\Tilde{\Phi}(u)w\\
            \text{subject to} \quad 
            &\alpha Aw\leq b- Au \label{eq:input_QP_cstr} \\
            &\alpha C\nabla h(u)w\leq d-Ch(u) \,, \label{eq:linearized_QP_cstr}
    \end{align} \label{eq:rewritten_linearized_QP}
\end{subequations}
where we have multiplied the objective with $\alpha$ and ignored the constant term in the objective. Since~\eqref{eq:rewritten_linearized_QP} is convex, the KKT conditions are necessary and sufficient to certify optimality of a solution $w$ of \eqref{eq:rewritten_linearized_QP}. Namely, $w \in \bbR^p$ is a solution if~\eqref{eq:input_QP_cstr}-\eqref{eq:linearized_QP_cstr} are satisfied and, for some dual multipliers $\nu \in\bbR^q_{\geq0}$ and $\mu \in\bbR^l_{\geq0}$, \emph{stationarity} \begin{equation}
\alpha w^TG(u)+\alpha\nabla\Tilde{\Phi}+\alpha\nu^T A+\alpha \mu^TC\nabla h(u) =0
    \label{eq:stationarity_KKT}
\end{equation} 
holds and \emph{complementary slackness} is satisfied, i.e., for all $j=1,...,q$ and all $i=1,...,l$, we have 
\begin{align} \label{eq:complementary_KKT}
    \nu_{j}(\alpha A_{j}w-b_{j}+A_{j}u)&=0\\
    \mu_{i}(\alpha C_{i}\nabla h(u)w-d_{i}+C_{i}h(u))&=0.
\end{align}

\begin{lemma}\label{aux_function_decrease}
Under \cref{ass:basic_feas,ass:linearizedLICQ,ass:basic}, let $V$ be as in~\eqref{eq:definition_V}, where $\xi$ is the upper bound of the Lagrange multipliers of the constraints in~\eqref{eq:lin_set_y} on $\calU$. Further, assume that $\nabla \tilde{\Phi}$ is $L$-Lipschitz and $C_i\nabla h$ is $\ell_i$-Lipschitz for all $i=1,...,l$ on $\calU$. For~\eqref{eq:outproj_approx}, $V(u^+)\leq V(u)$ is satisfied for all $u \in \mathcal{U}$, if 
\begin{equation}\label{eq:alpha_bound}
    \alpha < \alpha^\star := \tfrac{2\lambda_{\mathrm{min}}(G(u))}{L+\xi\textstyle\sum_{i=1}^l \ell_i}.
\end{equation}
\end{lemma}
\begin{proof}
In the following let $w := \sigma_\alpha(u)$. With the Descent Lemma (\cref{descent_lemma}), we can establish 
 \begin{align}\label{eq:taylor_phi}
    \Tilde{\Phi}(u^+)-\tilde{\Phi}(u)&\leq\alpha\nabla\Tilde{\Phi}(u)w+\tfrac{L}{2}||\alpha w||^2.
\end{align}
Further, from~\cref{lemma:transient_bound_violations}, we can derive
\begin{align}
    \mathrm{max}\{0, C_{i}h(u^+)-d_{i}\}&\leq \tfrac{\ell_i}{2}||\alpha w||^2 \,. \label{eq:taylor_li_max}
\end{align}
The following is inspired by the proof of~\cite[Lemma 10.4.1]{sherali2006nonlinear}. We take the inner product of~\eqref{eq:stationarity_KKT} and $w$, which results in $\alpha\nabla\Tilde{\Phi}(u)w=-\alpha w^TG(u)w-\textstyle\sum_{j=1}^{q}\alpha \nu_jA_jw-\textstyle\sum_{i=1}^{l}\alpha \mu_iC_i\nabla h(u)w$. Using~\eqref{eq:complementary_KKT}, we replace the summands, i.e., $\alpha\nabla\Tilde{\Phi}(u)w=-\alpha w^TG(u)w+\textstyle\sum_{j=1}^{q}\nu_{j}(A_{j}u-b_{j})+\textstyle\sum_{i=1}^{l}\mu_{i}(C_{i}h(u)-d_{i})$, which can be estimated as
\begin{multline}
    \alpha\nabla\Tilde{\Phi}(u)w 
    \leq -\alpha w^TG(u)w+\textstyle\sum_{i=1}^{l}\mu_{i}(C_{i}h(u)-d_{i})\\
    \leq -\alpha w^TG(u)w+\textstyle\sum_{i=1}^{l}\mu_{i}\max\{0,C_{i}h(u)-d_{i}\}.
    \label{eq:KKT_QP_estimated}
\end{multline}
Using \eqref{eq:taylor_phi}, \eqref{eq:taylor_li_max} and \eqref{eq:KKT_QP_estimated}, we obtain
\begin{align}\label{eq:V_descent}
    V(u^+)-V(u)&\leq-\alpha\lambda_{\mathrm{min}}(G(u))||w||^2\nonumber\\
    &+\tfrac{\alpha^2}{2} \left[ {L}+\xi\textstyle\sum_{i=1}^{l}\ell_i \right] ||w||^2\\
    &-\textstyle\sum_{i=1}^{l}(\xi-\mu_{i})\mathrm{max}\{0,C_{i}h(u)-d_{i}\}\nonumber.
\end{align}
Choosing $\alpha$ as in~\eqref{eq:alpha_bound} guarantees $V(u^+) \leq V(u)$.
\end{proof}

\subsubsection*{Convergence to first-order optimal points}
To show \ref{it:cvg} in \cref{MainTheorem_Linearized}, we can now simply apply \cref{thm:lasalle_discrete}. Namely, \cref{lem:closediteration,lemma_linearized_output_projection_continuity} guarantee that $T(u):=u+\alpha\,\sigma_\alpha(u)$ is continuous in $u$ and $\calU$ is invariant. By \cref{aux_function_decrease}, for the given $\alpha$ satisfying~\eqref{eq:alpha_bound}, the continuous function $V:\calU\rightarrow\bbR$ in~\eqref{eq:definition_V} is non-increasing along the trajectory of~\eqref{eq:outproj_approx} for all $u\in\calU$. The set $\calU$ being compact, we have $\mathrm{\mathbf{u}}=\{u_0,u_1,u_2,...\}\subset\calU$ being a bounded solution. Hence, for some $c\in V(\calU)$, the trajectory $\mathrm{\mathbf{u}}$ of~\eqref{eq:outproj_approx} converges to the largest invariant subset of $V^{-1}(c)\cap\{u\in\calU\,|\,V(u^+)-V(u)=0\}$.

Now, we show that \eqref{eq:outproj_approx} converges to the set of equilibrium points in $\Tilde{\mathcal{U}}$, that is $V(u^+)-V(u)=0$ implies that $u^{+}=u\in\Tilde{\mathcal{U}}$. For $V(u^+)-V(u)=0$, \eqref{eq:V_descent} reduces to
\begin{equation}
\begin{aligned}
    0&\leq \left(\tfrac{\alpha}{2} \left[ L+\xi\textstyle\sum_{i=1}^{l}\ell_i \right]- \lambda_{\mathrm{min}}(G(u)) \right) \alpha||w||^2\\
    &-\textstyle\sum_{i=1}^{l}(\xi-\mu_{i})\mathrm{max}\{0,C_{i}h(u)-d_i\}.
\end{aligned}
\label{eq:V_descent_equilibrium}
\end{equation}
For $\alpha$ as in~\eqref{eq:alpha_bound}, the right-hand side of \eqref{eq:V_descent_equilibrium} is negative for all $w\neq0$. It follows $w=0$ and $Ch(u)\leq d$, i.e., $u^{+}=u\in\Tilde{\mathcal{U}}$ is an equilibrium. Note that an uncertainty in $\nabla h$ does not affect the feasibility of equilibria (consider~\eqref{eq:rewritten_linearized_QP} for $w=0$).

Further, if $w^\star=0$ solves~\eqref{eq:rewritten_linearized_QP} at $u^\star$, and $\nu^\star, \mu^\star$ are the associated Lagrange multipliers, then the triplet $(u^\star, \nu^\star, \mu^\star)$ satisfies the first-order optimality conditions of {\eqref{eq:constrained_opt}}. In particular, $u^\star$ is feasible for~{\eqref{eq:constrained_opt}} and given LICQ\footnote{{\cref{ass:linearizedLICQ} implies LICQ of $\tilde{\calU}$ for all $u\in\tilde{\calU}$. This can be verified for $y=h(u)$, $w=0$ and $u\in\tilde{\calU}$.}}, the KKT conditions in \cref{theorem:KKT_condition_nonlinear_optimization_problem} are satisfied for~{\eqref{eq:constrained_opt}}.

\subsubsection*{Strict optimality of asymptotically stable equilibria}
For \ref{it:stab} in \cref{MainTheorem_Linearized}, the argumentation is similar to the proof of~\cite[Thm 5.5]{hauswirth2018projected}, albeit for a discrete-time system. Consider the neighborhood $\mathcal{N}(u^\star)\subset\mathcal{U}$ of $u^\star\in\tilde{\calU}$, such that any solution $\mathbf{u}$ to \eqref{eq:outproj_approx} starting at $u_0\in\mathcal{N}(u^\star)$ converges to $u^\star$. For the given $\alpha$, by \cref{aux_function_decrease}, we have $V(u_0)\geq V(u^\star)$, which implies either $\tilde{\Phi}(u_0)\geq\tilde{\Phi}(u^\star)$ or $u_0\not\in\tilde{\calU}$ or both. Hence, if $u_0\in\tilde{\calU}$, we have $V(u_0) = \Phi(u_0)$ and $\Tilde{\Phi}(u_0)\geq \Tilde{\Phi}(u^\star)$ follows. Since this reasoning applies to all $u_0\in\mathcal{N}(u^\star)$, it follows that $u^\star$ is a local minimizer of $\Tilde{\Phi}$ on $\Tilde{\mathcal{U}}$. 
To see that $u^\star$ is a strict local minimizer of $\Tilde{\Phi}$ on $\Tilde{\mathcal{U}}$, assume for the sake of contradiction that for an $\hat{u}\neq u^\star$ in the region of attraction $\mathcal{N}(u^\star)$ of $u^\star$, such that $\hat{u}\in\mathcal{N}(u^\star)\cap\Tilde{\mathcal{U}}$, it holds $\Tilde{\Phi}(\hat{u})=\Tilde{\Phi}(u^\star)$, and therefore (by feasibility) $V(\hat{u})=V(u^\star)$. Nevertheless, the solution $\mathbf{u}$ starting at $\hat{u}$ converges to $u^\star$ by assumption. Since for the given $\alpha$, by \cref{aux_function_decrease}, $V$ is non-increasing along the trajectory of~\eqref{eq:outproj_approx}, it follows $V(u^+)=V(u)$ for all iterates of the solution $\mathbf{u}$ starting at $\hat{u}$. However, as shown in the proof of \ref{it:cvg} in \cref{MainTheorem_Linearized}, $V(u^+)=V(u)$ implies that the point $u^+=u$ is an equilibrium point in $\tilde{\calU}$. Consequently, $u^\star$ cannot be asymptotically stable in $\mathcal{N}(u^\star)$.

\section{Numerical Example}\label{sec:numerical_simulation}
We illustrate the behavior of our control scheme in a small numerical example.
Namely, for $u\in\bbR^2$, $y\in\bbR$ and the map $y=h(u)=u_2^3+u_1-u_2+0.5$, we consider the minimization of $\Phi(u,y)=1.5u_1^2+u_2^2-u_2^3+u_1u_2-3u_2+1.5+y$ on $\calU:=[-1,1]^2$ and $\calY:=[0,1]$. To find a local solution, we use the feedback controller~\eqref{eq:feedbackupdate} {for $G\equiv \mathbb{I}_2$ and the exact measurement $y=h(u)$ of the system output. Implementations with other metrics (e.g., the Hessian metric, that yields a Newton flow) show a similar transient performance.}

\emph{Ease of tuning}: We start with a small fixed step-size $\alpha$ and gradually increase it. The left of Fig.~\ref{fig:planes} shows that all constraints are satisfied asymptotically and, depending on $\alpha$, {temporary} constraint violations can be made arbitrarily small (\cref{lemma:transient_bound_violations}). {Moreover, when considering the colored circles of the trajectories on the left of Fig.~\ref{fig:planes},} we can identify a trade-off between the magnitude of the {temporary} constraint violations and the convergence rate (in terms of iterations).

\emph{Comparison:} To highlight the benign convergence behavior of our approach, we compare it to a generic projected saddle-point scheme. Let the augmented Lagrangian of~\eqref{eq:constrained_opt}, when reformulated in the input coordinates, be defined as
\begin{align*}
    L(u,\mu)=\tilde{\Phi}(u)+\mu^T(Ch(u)-d)+\tfrac{\rho}{2}\max\{0,Ch(u)-d\}^2
\end{align*}
where $\rho\geq0$ is a fixed augmentation parameter. We consider the projected primal-dual scheme of the form
\begin{align*}
    u^+=P_\calU(u-\alpha\nabla_u L^T),\quad \mu^+=\max\{0,\mu+\gamma\nabla_\mu L^T\},
\end{align*}
where $P_\calU$ denotes the Euclidean projection on $\calU$, and $\alpha >0$ and $\gamma>0$ are separate primal and dual step-size parameters.

For the saddle-point simulations on the right of Fig.~\ref{fig:planes} and in Fig.~\ref{fig:objective_gamma_variation}, we have fixed $\alpha=0.01$ and vary $\gamma$ and $\rho$. Different values of $\rho$ and $\gamma$ result in different behaviors but with similar issues. For instance, a too large $\gamma$ does not guarantee convergence (see purple trajectory with $\gamma=5, \rho=1$). Further, a too large $\rho$ might jeopardize the stability of the controller (see orange trajectory with $\gamma=5, \rho=1000$). Hence, tuning with three parameters ($\alpha$, $\gamma$, $\rho$) can be challenging to manage various performance trade-offs. 

Moreover, the oscillatory nature of the trajectories reveals how saddle-point flows are not intended to satisfy constraints during transients. In contrast, our scheme demonstrates that the integration of the dual variable is not required to guarantee constraint satisfaction with zero asymptotic error, and it yields a benign convergence behavior without oscillations.

\begin{figure}[t] 
    	\centering
    	\includegraphics[scale=0.74]{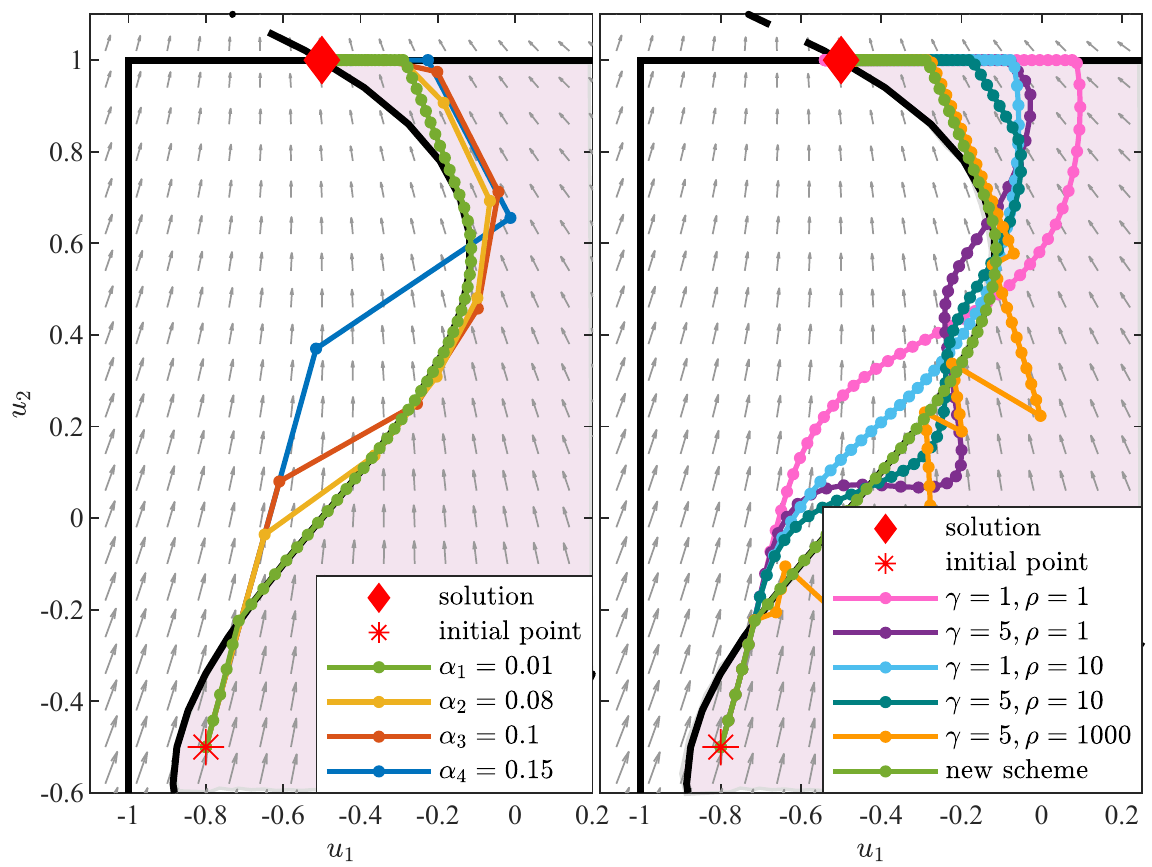}
    \caption{Trajectories in the $u_1u_2$-plane, where the vector field represents $-G(u)^{-1}\nabla\tilde{\Phi}(u)$. Left: new scheme for different $\alpha$. Right: (augmented) saddle-point scheme for {different $\gamma$ and $\rho$}. The colored patch represents $\tilde{\calU}$.}
    	\label{fig:planes}
\end{figure}
\begin{figure}[t] 
    	\centering
    	\includegraphics[scale=0.74]{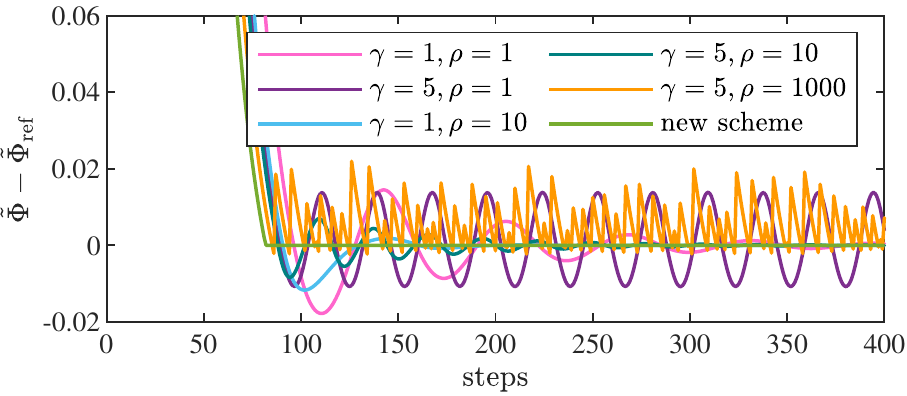}
    	\caption{Deviation of the objective value along the trajectories of the (augmented) saddle-point scheme for different $\gamma$ and $\rho$ from the reference value.}
    	\label{fig:objective_gamma_variation}
    	\vspace{-2mm}
\end{figure}

\section{Conclusions}\label{sec:conclusion}

We have proposed a new discrete-time controller for feedback optimization with the property that constraints on the steady-state plant outputs are enforced asymptotically. In contrast to saddle-point methods, the proposed scheme comes with global convergence guarantees for non-convex problems and fixed step-sizes, is easy to tune, and exhibits a well-behaved convergence behavior without oscillations. 

There is numerical and experimental evidence~\cite{ortmannExperimentalValidationFeedback2020} that the proposed method (and feedback optimization in general) is also very robust with respect to uncertainty in the input-output sensitivities.
A preliminary theoretical analysis of robustness is presented in~\cite{colombino2019towards}, where however no projection on the output constraint is considered. 

Ongoing research includes the relaxation of assumptions, the extension to more general constraint sets, the provision of convergence rates, and the study of the robustness in the output-constrained case. 

\appendix
In this appendix we discuss how the proposed discrete-time control law in~\eqref{eq:outproj_approx} is connected to a continuous-time projected gradient flow \cite{hauswirthProjectedgradientdescent2016,hauswirth2018projected,nagurneyProjectedDynamicalSystems1996}.

For this purpose, recall that the \emph{tangent cone} $T_u\tilde{\calU}$ of $\tilde{\calU}\subset\bbR^p$ at $u\in\tilde{\calU}$ is, informally speaking, the set of directions in which one can leave $u$ and remain in the set $\tilde{\calU}$. If $\tilde{\calU}$ takes the form in~\eqref{eq:tilde_U} and satisfies LICQ for all $u\in\tilde{\calU}$, then the \emph{tangent cone of $\tilde{\calU}$ at $u\in\tilde{\calU}$} is given by
    \begin{align}\label{eq:tangent_cone}
      T_u\tilde{\calU}:=\left\{w\in\mathbb{R}^p\,\middle|\,A_{\mathbf{I}^\calU_u}w\leq 0\,, C_{\mathbf{I}^\calY_{h(u)}} \nabla h(u)w\leq 0\right\}.
    \end{align}

Given $\tilde{\Phi}:\bbR^p\rightarrow\bbR$, a \emph{projected gradient flow for $\tilde{\Phi}$ on $\tilde{\calU}$} in \cite{hauswirthProjectedgradientdescent2016,hauswirth2018projected,nagurneyProjectedDynamicalSystems1996} is defined as
\begin{align}\label{eq:proj_grad}
    \dot u = \Pi_{\tilde{\calU}}^G \big[ -G^{-1}(u)\nabla\tilde{\Phi}(u)^T \big] (u),\,\,\,\,u\in\tilde{\calU},
\end{align}
where $\Pi^G_{\tilde{\calU}}[f](u)$ projects a vector field $f:\tilde{\calU}\rightarrow\bbR^p$ onto the tangent cone $T_u \tilde{\calU}$ of $\tilde{\calU}$ at the point $u\in\tilde{\calU}$, i.e.,
\begin{align}\label{eq:proj_operator} 
    \Pi^G_{\tilde{\calU}}[f](u) :=\arg \min_{w\in T_u\tilde{\calU}} ||w-f(u)||^2_{G(u)}\,.
\end{align}
\vspace{-1mm}
In particular, the gradient of $\tilde{\Phi}$ is projected in such a way that solution trajectories cannot leave the set $\tilde{\calU}$. 

A more rigorous treatment, including requirements for the existence and convergence of solutions for~\eqref{eq:proj_grad} can be found in~\cite{hauswirth2018projected, hauswirthProjectedgradientdescent2016, nagurneyProjectedDynamicalSystems1996} and references therein.

We can establish that $\sigma_{\alpha}$ converges to $\Pi^G_{\tilde{\calU}}$ as $\alpha\searrow0^+$:

\begin{lemma}\label{lemma:dt_to_ct}
Consider~\eqref{eq:outproj_approx} and~\eqref{eq:sigma_operator}, and let~\cref{ass:basic_feas,ass:linearizedLICQ,ass:basic} be satisfied. Then, for all $u \in \tilde{\calU}$, we have
\begin{align}\label{eq:lim_alpha_zero}
    \underset{\alpha \searrow 0^+}{\lim} \, \sigma_{\alpha}(u) = \Pi_{\tilde{\calU}}^G[-G^{-1}(u)\nabla\tilde{\Phi}(u)^T](u) \, .
\end{align}
\end{lemma}

\begin{proof}
The equivalence in~\eqref{eq:lim_alpha_zero} can be proven by considering~\eqref{eq:sigma_operator}, where the constraints are separated based on whether they are active or inactive constraints of $u$ in~\eqref{eq:constrained_opt}, i.e.,
\begin{subequations}
\begin{align} \label{eq:objective_rearrange_linearized_output_discretization_d}
           { \underset{w\in\bbR^p}{\arg \min} }\quad 
            &{\left\| w+G^{-1}(u)\nabla\tilde{\Phi}(u)^T \right\| _{G(u)}^2}\\\label{eq:active_constraints1}
            {\text{subject to}} \quad  & {A_{\mathbf{I}^\calU_u}w\leq 0} \\ \label{eq:active_constraints2}
            &{C_{\mathbf{I}^\calY_{h(u)}} \nabla h(u) w\leq 0}\\  \label{eq:inactive_constraints1}
            &{A_{\Bar{\mathbf{I}}^\calU_u} w\leq \tfrac{1}{\alpha}(b-A_{\Bar{\mathbf{I}}^\calU_u}u)}\\ \label{eq:inactive_constraints2}
            &{C_{\Bar{\mathbf{I}}^\calY_{h(u)}}\nabla h(u) w\leq \tfrac{1}{\alpha}(d-C_{\Bar{\mathbf{I}}^\calY_{h(u)}}h(u)).}
\end{align}
\label{eq:sigma_operator_rearrange}
\end{subequations}
For $\alpha\searrow0^+$ and $u\in\tilde{\calU}$,~\eqref{eq:inactive_constraints1} and \eqref{eq:inactive_constraints2} can be omitted, since their right hand sides tend to $+\infty$. Hence, the constraint set can be reduced to~\eqref{eq:active_constraints1} and~\eqref{eq:active_constraints2}, denoting $T_u\tilde{\calU}$ in~\eqref{eq:tangent_cone}, and~\eqref{eq:objective_rearrange_linearized_output_discretization_d} - \eqref{eq:active_constraints2} is equivalent to the evaluation of~\eqref{eq:proj_grad}.\end{proof}
\cref{lemma:dt_to_ct} offers an initial insight on how \eqref{eq:outproj_approx} and \eqref{eq:proj_grad} are connected, however, it does not make a statement about the \emph{uniform approximation} of solutions of \eqref{eq:proj_grad} by solutions of \eqref{eq:outproj_approx} as $\alpha \searrow 0$. This conjecture remains an open problem.

\bibliographystyle{IEEEtran}
\bibliography{IEEEabrv,bibliography} 

\end{document}